\documentclass{amsart}[12pt]
\usepackage[english]{babel}
\usepackage{amssymb}

\usepackage[utf8]{inputenc}
\usepackage[letterpaper,top=2cm,bottom=2cm,left=3cm,right=3cm,marginparwidth=1.75cm]{geometry}

\usepackage{amsthm}
\usepackage{amsmath}
\usepackage{graphicx}

\usepackage{xcolor}
\usepackage[colorlinks=true, allcolors=blue]{hyperref}
\newcommand{\C}{\mathbb C}

\newcommand{\R}{\mathbb R}
\newcommand{\K}{\mathbb K}
\renewcommand{\H}{\mathbb H}
\renewcommand{\S}{\mathbb S}
\newcommand{\ket}[1]{\left|#1\right\rangle}
\newcommand{\bra}[1]{\left\langle#1\right|}
\newcommand{\proj}[1]{\ket{#1}\!\bra{#1}}
\newcommand{\Tr}{\operatorname{Tr}}

\newtheorem{theorem}{Theorem}[section]
\newtheorem{lemma}[theorem]{Lemma}
\newtheorem{proposition}[theorem]{Proposition}
\newtheorem{corollary}[theorem]{Corollary}
\newtheorem{definition}[theorem]{Definition}
\newtheorem{remark}[theorem]{Remark}
\newtheorem{example}[theorem]{Example}

\newtheorem{cnst}[theorem]{Construction}

\DeclareMathOperator{\supp}{Supp}

\DeclareMathOperator{\GL}{GL}

\title{Pure-State Quantum Tomography with Minimal Rank-One POVMs}
\author{Dan Edidin, Ivan Gonzalez, and Itzhak Tamo}
\begin{document}

\begin{abstract}
Quantum state tomography seeks to reconstruct an unknown state from measurement statistics. A finite measurement (POVM) is \emph{pure-state informationally complete} (PSI-Complete) if the outcome probabilities determine any pure state up to a global phase. We study \emph{rank-one} POVMs that are minimally sufficient for this task. We call such a POVM \emph{vital} if it is PSI-Complete but every proper subcollection is not PSI-Complete.

We prove sharp upper bounds on the size of vital rank-one POVMs in dimension \(n\): the size is at most \(\binom{n+1}{2}\) over \(\mathbb{R}\) and at most \(n^{2}\) over \(\mathbb{C}\), and we give constructions that attain these bounds. In the real case, we further exhibit a connection to block designs: whenever \(w \mid n(n-1)\), an \((n,w,w-1)\) design produces a vital rank-one POVM with \(n + n(n-1)/w\) outcomes. We provide explicit constructions for \(w=2,n-1\), and \(n\).
\end{abstract}
\maketitle

\section{Introduction}

The reconstruction of an unknown quantum state from experimental data, known as Quantum State Tomography (QST), is a fundamental problem in quantum information science. In a finite-dimensional Hilbert space $\K^n$ (where $\K=\R$ or $\C$), a measurement with $m$ outcomes is described by a Positive Operator-Valued Measure (POVM), which is a set of positive semidefinite operators $\{E_i\}_{i=1}^m$ such that $\sum_{i=1}^m E_i = I$. When the system is in a pure state $\rho=\lvert\psi\rangle\langle\psi\rvert$, the Born rule dictates that the probability of the $i$-th outcome is $p_i(\psi)=\operatorname{Tr}(\rho E_i)=\langle\psi,E_i\psi\rangle$.

A central question is whether the measurement statistics $\{p_i(\psi)\}$ are sufficient to uniquely identify the state $\psi$ (up to the inherent global phase ambiguity). If so, the POVM is called \emph{pure-state informationally complete} (PSI-Complete).

We focus on rank-one POVMs, where each operator $E_i$ has rank 1. In this case, $E_i =  v_i v_i^*$ for some vector $v_i \in \K^n$. The PSI-Complete condition is then equivalent to the injectivity of the map $\K^n/\K^\times \to \R^m$ given by $\psi \mapsto (|\langle v_i,\psi\rangle|^2)$. This mathematical problem is extensively studied in applied harmonic analysis and signal processing under the name of \emph{phase retrieval} \cite{eldar,candes,Fienup1982}.

The property of PSI-Complete POVM depends on the geometric configuration of the subspaces spanned by the vectors $\{v_i\}$. A frame $\mathcal{F}=\{v_i\}\subset\K^n$, which is a spanning set of vectors, can be canonically transformed into a rank-one POVM via a whitening process, as follows. Define the frame operator $S=\sum v_i v_i^*$, the vectors $w_i = S^{-1/2}v_i$ form a Parseval frame, and $\{E_i=w_i w_i^*\}$ constitutes a POVM. Importantly,  the PSI-Complete property is preserved under whitening; hence,  a frame $\{v_i\}$ allows for phase retrieval if and only if the associated POVM $\{w_i w_i^*\}$ is PSI-Complete. 
Consequently, we use the terminology and results of frame theory to study rank-one PSI-Complete POVMs, and we will use the terms `vital frame' and `vital rank-one POVM' interchangeably, as defined later.

The minimal number of measurements required for PSI-Complete has been a subject of significant research. It is known that a generic frame of size $m \geq 2n-1$ suffices in the real case~\cite{balan2006signal}, and $m \geq 4n-4$ suffices in the complex case~\cite{conca2015algebraic}. The bound is sharp in the real case, and if $n = 2^k +1$
in the complex case, although in the complex case exceptions exist, such as the existence of 11-element frames in $\C^4$ that permit phase retrieval~\cite{vinzant2015small}.

While these results establish the minimum necessary size, verifying whether a specific structured frame of size $O(n)$ achieves phase retrieval is computationally difficult. By contrast, practical QST often employs highly structured, larger measurements of size $O(n^2)$, such as those derived from Mutually Unbiased Bases~\cite{WoottersFields1989,balan2009painless} or certain Gabor systems~\cite{bojarovska2015}. These larger measurements offer advantages for verification and reconstruction but are redundant for pure-state identification.

We ask which measurements are exactly sufficient, i.e., PSI-complete with no redundancy. 
A POVM is said to be \emph{vital} if it is PSI-complete, but the removal of any single measurement operator $E_j$ (followed by renormalization of the remaining elements) destroys the PSI-complete property. 
Note that this definition applies to arbitrary POVMs, not only to rank-one POVMs.

By definition, generic frames of the minimal size ($2n-1$ or $4n-4$) are vital. Our interest lies in characterizing vital rank-one POVMs beyond these minimal sizes. The second author~\cite{gonzalez2024thesis} previously gave the first examples of real vital frames of size $m > 2n-1$. Here, we provide bounds and new constructions for both the real and complex cases.

\subsection{Motivation for the Real Case (Rebits) and Combinatorial Designs}
While standard quantum mechanics is formulated over $\mathbb{C}$, scenarios exist where states and measurements are effectively restricted to $\mathbb{R}$ relative to a fixed basis. This occurs, for instance, in systems exhibiting certain symmetries (e.g., time-reversal symmetry) or in the study of ``Rebits''. The real case is mathematically distinct, as evidenced by the different minimal bounds. Moreover, in the real case, PSI-Completeness (or phase retrieval) is equivalent to the fact that the associated frame admits the  Complement Property \cite{balan2006signal} (a frame in $\R^n$ has phase retrieval if and only if for every partition of the frame into two subsets, at least one of them spans $\R^n$). We exploit this property to connect the existence of structured vital rank-one POVMs to combinatorial designs. 
Block designs provide a framework for constructing measurements with sparse, structured incidence properties, which allow us to rigorously certify both the PSI-Complete property and vitality.

\subsection{Contributions}
The main contributions of this paper are as follows:
\begin{itemize}
    \item We establish sharp upper bounds on the size of vital rank-one POVMs. We prove that any vital measurement in $\R^n$ has at most $\binom{n+1}{2}$ outcomes, and in $\C^n$ has at most $n^2$ outcomes.
    \item We provide explicit constructions attaining these bounds. For example, in the real case, the measurement derived from $\{\lvert e_i\rangle, \lvert e_i{+}e_j\rangle\}_{i<j}$ is vital. In the complex case, adding the vectors $\lvert e_i{+}ie_j\rangle$ yields a vital measurement with $n^2$ outcomes.
    \item We introduce a general method for constructing a large family of  vital rank-one POVMs in $\R^n$ using combinatorial block designs. Specifically, we prove that for $w\mid n(n-1)$, any $(n,w,w\!-\!1)$ block design gives rise to a vital rank-one POVM  with $n+n(n-1)/w$ outcomes. This construction unifies known examples and provides new ones.
    \item We show that the families constructed via block designs do not exhaust the space of vital frames, showing the existence of vital frames of size $2n$ (for $n>3$) that are not equivalent to those arising from the corresponding block design construction.
\end{itemize}

\subsection{Organization}
The remainder of the paper is organized as follows. 
Section~\ref{preliminaries}
provides some preliminaries for Quantum State Tomography and defines the key concepts used in the paper. 
In Section \ref{sec:upper-bounds}, we prove the upper bounds on the size of vital rank-one POVMs in both the real and complex cases and provide explicit constructions showing that these bounds are sharp. Section~\ref{sec.design} provides the construction of vital rank-one POVMs from block designs in the real case, establishing the connection between the combinatorial properties of the design and the vitality of the resulting frame. In Section \ref{sec:equiv-vital-measurements}, we show that the block design construction does not provide all non-equivalent classes of vital rank-one POVMs. 
We conclude in Section \ref{sec:conclusions} with some open questions.
\section{Preliminaries}
\label{preliminaries}
In this section, we establish the mathematical framework for Quantum State Tomography (QST) and define the key concepts used throughout the paper.

\subsection{State Space and Operators}
Let $\K$ denote either the real field $\R$ or the complex field $\C$. We consider an $n$-dimensional Hilbert space $\K^n$.

We denote the space of real symmetric $n \times n$ matrices by $\S_n$ (when $\K=\R$), and the space of Hermitian $n \times n$ matrices by $\H_n$ (when $\K=\C$). Both are real vector spaces, equipped with the Hilbert-Schmidt inner product $\langle A, B \rangle = \operatorname{Tr}(A^*B)$. The dimensions of these spaces are $\dim \S_n = \binom{n+1}{2}$ and $\dim \H_n = n^2$.

A quantum state is described by a density matrix $\rho$, which is a positive semidefinite operator ($\rho \succeq 0$) with $\operatorname{Tr}(\rho)=1$. A state is \emph{pure} if $\rho$ has rank 1, in which case it can be written as $\rho = \lvert\psi\rangle\langle\psi\rvert$ (denoted $\operatorname{proj}(\psi)$) for some unit vector $\psi \in \K^n$. The vector $\psi$ is determined uniquely up to a global phase. If the rank of $\rho$ is greater than 1, the state is \emph{mixed}.

\subsection{Measurements and the Measurement Map}
\label{subsec:measurement-map}
A finite-outcome measurement is modeled by a positive operator-valued measure (POVM) as follows.

\begin{definition}[POVM, rank-one POVM]
A POVM on $\K^n$ is a finite set $\mathcal{E}=\{E_i\}_{i=1}^m$ of operators such that $E_i\succeq0$ for all $i$ and $\sum_{i=1}^m E_i=I$. The POVM is \emph{rank-one} if $\operatorname{rank}(E_i)=1$ for all $i$.
\end{definition}

If $\mathcal{E}$ is a rank-one POVM, each $E_i$ can be written as $E_i = v_i v_i^*$ for some vector $v_i \in \K^n$. In this case, it is easy to verify that the set of vectors $\mathcal{F}=\{v_i\}_{i=1}^m$ spans $\K^n$ and is called a frame; since $\sum v_i v_i^* = I$, it is specifically a Parseval frame, i.e., $\sum_i |\langle x,v_i\rangle|^2=\|x\|
^2$ for all $x\in \K^n.$

When a state $\rho$ is measured using a POVM $\mathcal{E}$, the probability of the $i$-th outcome is given by the Born rule: $p_i = \operatorname{Tr}(\rho E_i)$. Following~\cite{bandeira2013saving}, we define the associated linear measurement map $\Phi_{\mathcal{E}}$ as follows. When $\K=\R$, the map is defined as
\begin{equation} \label{eq.real_map}
\Phi_{\mathcal{E}} \colon \S_{n} \to \R^m; \quad A \mapsto \left( \langle A, E_1 \rangle, \ldots, \langle A, E_m \rangle \right),
\end{equation}
and when $\K=\C$, it is defined as
\begin{equation} \label{eq.complex_map}
\Phi_{\mathcal{E}} \colon \H_{n} \to \R^m; \quad A \mapsto \left( \langle A, E_1 \rangle, \ldots, \langle A, E_m \rangle \right).
\end{equation}
The goal of QST is to reconstruct the state $\rho$ from the measurement outcomes $\Phi_{\mathcal{E}}(\rho)$.

\subsection{Informational Completeness and Vitality}

The ability to reconstruct states depends on the injectivity of the measurement map $\Phi_{\mathcal{E}}$.

\begin{definition}[Informational Completeness (IC)]
A POVM $\mathcal{E}$ is \emph{Informationally Complete} (IC) if the map $\Phi_{\mathcal{E}}$ is injective on the entire space $\S_n$ (resp. $\H_n$).
\end{definition}

An IC POVM allows for the reconstruction of any state, including mixed states. This requires that the operators $\{E_i\}$ span the entire operator space. Consequently, a necessary condition for IC is $m \geq \binom{n+1}{2}$ in the real case or $m \geq n^2$ in the complex case.

In this paper, we are interested in the scenario in which the state is known a priori to be pure. Next, we define when a measurement map is informationally complete for the set of pure states.

\begin{definition}[Pure-State Informationally Complete (PSI-Complete)]
A POVM $\mathcal{E}$ is PSI-Complete if for all unit vectors $\psi,\phi \in \K^n$, the condition
\[
\Phi_{\mathcal{E}}(\operatorname{proj}(\psi)) = \Phi_{\mathcal{E}}(\operatorname{proj}(\phi))
\]
implies that $\operatorname{proj}(\psi) = \operatorname{proj}(\phi)$. In other words, $\Phi_\mathcal{E}$ is injective when restricted to the set $\{\operatorname{proj}(\psi):\psi\in \K^n, \|\psi\|=1 \}$. 
\end{definition}

By linearity, $\mathcal{E}$ is PSI-Complete if and only if there are no nonzero operators of the form $\operatorname{proj}(\psi) - \operatorname{proj}(\phi)$ in the kernel of $\Phi_{\mathcal{E}}$. Such operators are indefinite and have rank   $2$. Lastly, note that IC implies PSI-Complete.

\subsection{Conditioning on \texorpdfstring{$\neg j$}{not j} and Vitality}
Let $\mathcal{E}$ be a POVM, fix $j$ and set $S_j:=\sum_{i\neq j}E_i=I-E_j$. For a pure state $\rho$, 
the probability for outcome $i$ is 
$p_\rho(i)=\Tr(\rho E_i)$. Assume that we discard the measurements \emph{with} outcome $j$,  then the conditional probabilities given “not $j$” are
\[
p_\rho(i\,|\,\neg j)=\frac{\Tr(\rho E_i)}{\Tr(\rho S_j)}\qquad (i\neq j),
\]
which are defined exactly on the domain
\[
\mathcal{D}_j\ :=\ \{\ \rho\ \text{pure} \ :\ \Tr(\rho S_j)>0\ \}.
\]
We say that the subcollection $\{E_i\}_{i\neq j}$ is PSI-Complete when the map
\[
\widetilde\Phi_{\mathcal{E},j}:\ \mathcal{D}_j\ \longrightarrow\ \R^{m-1},\quad
\rho\ \longmapsto\ \big(p_\rho(i\,|\,\neg j)\big)_{i\neq j}
\]
is injective on $\mathcal{D}_j$.

\medskip
The following proposition shows that having a trivial kernel for $S_j$ is necessary for the subcollection $\{E_i\}_{i\neq j}$ to be PSI-Complete. Furthermore, under this assumption, it gives a necessary and sufficient condition for the subcollection to be PSI-Complete.
\begin{proposition}
\label{prop:good_prop}
\begin{enumerate}
    \item 
If the subcollection $\{E_i\}_{i\neq j}$ is PSI-Complete (i.e., $\widetilde\Phi_{\mathcal{E},j}$ is injective on $\mathcal D_j$), then $S_j$ has trivial kernel and $\mathcal D_j$ is the set of all pure states.
\item Assume that $S_j$ has a trivial kernel. Then the subcollection $\{E_i\}_{i\neq j}$ is PSI-Complete if and only if there do not exist distinct pure states $\rho,\sigma\in\mathcal D_j$ and $\alpha>0$ with
\[
\Tr(\rho E_i)=\alpha\,\Tr(\sigma E_i)\quad\text{for all }i\neq j.
\]
\end{enumerate}
\end{proposition}

\begin{proof}(1)  Assume towards a contradiction that $\ker S_j\neq\{0\}$. Pick a nonzero $k\in\ker S_j$. Then $$0=\langle k,S_jk\rangle=\sum_{i\neq j}\langle k,E_i k\rangle.$$ Since each $E_i\succeq0$, we have $E_i k=0$ for all $i\neq j$. Since $S_j\neq 0$, there exists a unit vector $u$ with $\Tr(\proj{u}S_j)>0$. Note that $u$ and $k$ are linearly independent since $k\in\ker S_j$ and $u\notin\ker S_j$. Define $N_t := \|u+tk\|$ and the unit vector $\psi_t=(u+tk)/N_t$. For $i\neq j$, we calculate the probabilities. Since $E_i k=0$ and $E_i$ is Hermitian (which implies $\langle k, E_i u \rangle = \langle E_i k, u \rangle = 0$), the cross terms vanish:$$\Tr(\proj{\psi_t}E_i)=\frac{\langle u+tk, E_i(u+tk) \rangle}{N_t^2} = \frac{\langle u, E_i u \rangle}{N_t^2} = \frac{\Tr(\proj{u}E_i)}{N_t^2}.$$Similarly, since $S_j k=0$ and $S_j$ is Hermitian,$$\Tr(\proj{\psi_t}S_j)=\frac{\Tr(\proj{u}S_j)}{N_t^2}.$$ Thus the conditional probabilities$$p_{\proj{\psi_t}}(i\,|\,\neg j)=\frac{\Tr(\proj{u}E_i)/N_t^2}{\Tr(\proj{u}S_j)/N_t^2} = \frac{\Tr(\proj{u}E_i)}{\Tr(\proj{u}S_j)}$$are independent of $t$. Since $u$ and $k$ are linearly independent, $\proj{\psi_t}$ varies with $t$. Hence $\widetilde\Phi_{\mathcal{E},j}$ is not injective on $\mathcal D_j$, a contradiction. Therefore, $S_j$ has a trivial kernel. This implies $S_j\succ 0$, so for every pure state $\rho=\proj{\psi}$ we have$\Tr(\rho S_j)=\langle\psi,S_j\psi\rangle>0$. Hence $\mathcal D_j$ equals the set of all pure states.

\smallskip
(2) \emph{($\Rightarrow$)} Suppose $\{E_i\}_{i\neq j}$ is PSI-Complete on $\mathcal D_j$. If there existed  pure states $\rho,\sigma\in\mathcal D_j$ and $\alpha>0$ with
$\Tr(\rho E_i)=\alpha\,\Tr(\sigma E_i)$ for all $i\neq j$, then summing over $i\neq j$ yields
$\Tr(\rho S_j)=\alpha\,\Tr(\sigma S_j)>0$, and dividing componentwise gives
\[
p_\rho(i\,|\,\neg j)
=\frac{\Tr(\rho E_i)}{\Tr(\rho S_j)}
=\frac{\Tr(\sigma E_i)}{\Tr(\sigma S_j)}
=p_\sigma(i\,|\,\neg j)\quad\forall\,i\neq j,
\]
which contradicts the injectivity of $\widetilde\Phi_{\mathcal{E},j}$.

\emph{($\Leftarrow$)} Conversely, assume there do \emph{not} exist distinct pure states \(\rho,\sigma\in\mathcal D_j\) and \(\alpha>0\) with \(\Tr(\rho E_i)=\alpha\,\Tr(\sigma E_i)\) for all \(i\neq j\). If $\{E_i\}_{i\neq j}$ were not PSI-Complete, there would exist distinct $\rho,\sigma\in\mathcal D_j$ such that
\[
\frac{\Tr(\rho E_i)}{\Tr(\rho S_j)}=\frac{\Tr(\sigma E_i)}{\Tr(\sigma S_j)}\quad\forall\,i\neq j.
\]
Let \(\alpha:=\Tr(\rho S_j)/\Tr(\sigma S_j)>0\). Multiplying both sides by \(\Tr(\rho S_j)\) gives
\[
\Tr(\rho E_i)=\alpha\,\Tr(\sigma E_i)\quad\forall\,i\neq j,
\]
contradicting the assumption. Therefore $\widetilde\Phi_{\mathcal{E},j}$ is injective.
\end{proof}

\medskip
The following lemma provides a simple necessary and sufficient condition for the subcollection $\{E_i\}_{i\neq j}$ to be PSI-Complete, phrased entirely in terms of unnormalized probabilities.

\begin{lemma}
\label{lem:psic-unnormalized-criterion}
The subcollection $\{E_i\}_{i\neq j}$ is PSI-Complete if and only if 
\begin{enumerate}
\item $S_j$ is invertible, and
\item for any two pure states $\rho\neq\sigma$, their images under the unnormalized map
\[
\Phi_{\mathcal{E}_{\setminus j}}:\ \rho\ \longmapsto\ \big(\Tr(\rho E_i)\big)_{i\neq j},
\]
are not positive scalar multiples of one another.
\end{enumerate}
\end{lemma}

\begin{proof}
(\(\Rightarrow\)) If the subcollection is PSI-Complete, then by Proposition \ref{prop:good_prop} (1) we must have $\ker S_j=\{0\}$, so $S_j$ is invertible and $\mathcal D_j$ is all pure states. If there existed distinct pure states $\rho\neq\sigma$ and $\alpha>0$ with $\Tr(\rho E_i)=\alpha\,\Tr(\sigma E_i)$ for all $i\neq j$, then  we would get $p_\rho(i\,|\,\neg j)=p_\sigma(i\,|\,\neg j)$ for all $i\neq j$, contradicting PSI-Complete.

(\(\Leftarrow\)) Assume $S_j$ is invertible and no proportional pair exists. Then $\mathcal D_j$ is all pure states. If the subcollection were not PSI-Complete, there would be distinct pure states $\rho,\sigma$ with
\[
\frac{\Tr(\rho E_i)}{\Tr(\rho S_j)}=\frac{\Tr(\sigma E_i)}{\Tr(\sigma S_j)}\quad\forall\,i\neq j.
\]
Setting $\alpha:=\Tr(\rho S_j)/\Tr(\sigma S_j)>0$ gives $\Tr(\rho E_i)=\alpha\,\Tr(\sigma E_i)$ for all $i\neq j$, contradicting the hypothesis. Hence, the subcollection is PSI-Complete.
\end{proof}
Next, we define the central concept of this paper: measurements that are minimally sufficient for pure-state tomography.

\begin{definition}[Vital Rank-One POVM]
A rank-one POVM $\mathcal{E}=\{E_i\}_{i=1}^m$ on $\K^n$ is \emph{vital} if it is PSI-Complete, but for every $j\in\{1,\ldots,m\}$ 
 the subcollection $\{E_i\}_{i\neq j}$ is not PSI-Complete.
\end{definition}
The following corollary, which follows easily from Lemma \ref{lem:psic-unnormalized-criterion}, provides a simple necessary and sufficient condition for a rank-one POVM to be vital and PSI-Complete.
\begin{corollary}
\label{stam_cor}A rank-one POVM $\mathcal{E}=\{E_i\}_{i=1}^m$ on $\K^n$ is vital if and only if it is PSI-Complete and for each $j$,  $S_j=I-E_j$ is singular or there exist distinct pure states whose unnormalized vectors $\big(\Tr(\rho E_i)\big)_{i\neq j}$ are proportional.
\end{corollary}

\subsection{Frames, Phase Retrieval, and Vitality}
A finite set of vectors $\mathcal{F}=\{v_i\}_{i=1}^m \subset \K^n$ is called a frame if it spans $\K^n$. The associated frame operator is $S=\sum v_i v_i^*$. If $S=I$, the frame is a Parseval frame. As noted in Section \ref{subsec:measurement-map}, a rank-one POVM corresponds exactly to a Parseval frame.
The problem of reconstructing a pure state $\psi$ from the measurements $\{|\langle v_i, \psi \rangle|^2\}_{i=1}^m$ generated by a frame $\mathcal{F}$ is known as phase retrieval. We say a frame $\mathcal{F}$ permits phase retrieval if the map induced by $ \mathcal {F} $ is injective (up to global phase).
Any frame $\mathcal{F}$ can be associated with a rank-one POVM via a whitening process: the vectors $w_i = S^{-1/2}v_i$ form the canonical Parseval frame, and $\{E_i = w_i w_i^*\}$ is the associated POVM.
\begin{remark}    
[Equivalence of Properties] A crucial observation is that a frame $\mathcal{F}$ permits phase retrieval if and only if its associated POVM is PSI-Complete.
Furthermore, the definition of vitality extends naturally to frames: a frame is vital if it permits phase retrieval, but removing any vector destroys this property. This property is also preserved under the whitening transformation. Therefore, throughout this paper, we often analyze the properties of the underlying frames directly, with the understanding that the results apply directly to the corresponding normalized POVMs.\end{remark}
\section{Bounds on the Size of vital rank-one POVMs}
\label{sec:upper-bounds}
In this section, we establish the maximum possible size of a vital rank-one POVM in both real and complex Hilbert spaces and provide explicit constructions that achieve these bounds.

\begin{theorem} \label{thm.bounds}
Any vital rank-one POVM in $\R^n$ has size at most $\binom{n+1}{2}$, and any vital rank-one POVM in $\C^n$ has size at most $n^2$. Moreover, these bounds are sharp for all $n$.
\end{theorem}

We first prove the upper bounds using the properties of the measurement map and the characterization of vitality defined in Section \ref{preliminaries}, and then show that the bounds are sharp via explicit constructions.

\subsection{Proof of the Upper Bounds}

Let $\mathcal{E} = \{E_1, \ldots, E_m\}$ be a vital rank-one POVM in $\K^n$. Let $\Phi_{\mathcal{E}}$ be the associated measurement map acting on $\S_n$ (if $\K=\R$) or $\H_n$ (if $\K=\C$). We aim to construct a set of $m$ linearly independent operators $\{A_1, \ldots, A_m\}$ in the corresponding operator space.

Since $\mathcal{E}$ is vital, by definition, it is PSI-Complete, and for every $j \in \{1, \ldots, m\}$, the subcollection $\mathcal{E}_{\setminus j} = \{E_i\}_{i \neq j}$ is not PSI-Complete. We utilize the characterization provided by Corollary~\ref{stam_cor}. For each $j$, one of the following two cases must hold:

\medskip
\noindent\textbf{Case 1: $S_j=I-E_j$ is singular.}
Since $E_j$ is a rank-one operator, $S_j$ is singular if and only if $1$ is an eigenvalue of $E_j$. Let $k_j$ be a corresponding unit eigenvector, so $E_j k_j = k_j$. This implies $S_j k_j = (I-E_j)k_j = 0$.
We define the operator $A_j = \operatorname{proj}(k_j)$. $A_j$ belongs to the operator space ($\S_n$ or $\H_n$).

We compute the inner products of $A_j$ with the POVM elements. First,
\begin{equation}
\langle A_j, E_j \rangle = \operatorname{Tr}(\operatorname{proj}(k_j) E_j) = \langle k_j, E_j k_j \rangle = \langle k_j, k_j \rangle = 1 \neq 0.
\end{equation}
Next, consider $i \neq j$. Since $S_j = \sum_{i\neq j} E_i$, then,
\[
\sum_{i\neq j} \langle A_j, E_i \rangle = \sum_{i\neq j} \langle k_j, E_i k_j \rangle = \langle k_j, S_j k_j \rangle = 0.
\]
Since $E_i \succeq 0$, we must have $\langle k_j, E_i k_j \rangle \geq 0$. Therefore,
\begin{equation}
\langle A_j, E_i \rangle = 0 \quad \text{for all } i \neq j.
\end{equation}

\medskip
\noindent\textbf{Case 2: $S_j$ is invertible, but there exist distinct pure states $\rho_j \neq \sigma_j$ and a scalar $\alpha_j > 0$ such that $\Phi_{\mathcal{E}_{\setminus j}}(\rho_j) = \alpha_j \Phi_{\mathcal{E}_{\setminus j}}(\sigma_j)$.}
This means $\operatorname{Tr}(\rho_j E_i) = \alpha_j \operatorname{Tr}(\sigma_j E_i)$ for all $i \neq j$.

We define the operator $A_j = \rho_j - \alpha_j \sigma_j$. $A_j$ is in the operator space.

Again, we compute the inner products with the POVM elements. For $i \neq j$,
\begin{equation}
\langle A_j, E_i \rangle = \operatorname{Tr}(A_j E_i) = \operatorname{Tr}(\rho_j E_i) - \alpha_j \operatorname{Tr}(\sigma_j E_i) = 0.
\end{equation}

Now consider the inner product with $E_j$. We use the fact that $\sum_i E_i = I$ and $\operatorname{Tr}(\rho_j)=\operatorname{Tr}(\sigma_j)=1$. Summing the proportionality relation over $i\neq j$:
\[
\operatorname{Tr}(\rho_j S_j) = \sum_{i\neq j} \operatorname{Tr}(\rho_j E_i) = \alpha_j \sum_{i\neq j} \operatorname{Tr}(\sigma_j E_i) = \alpha_j \operatorname{Tr}(\sigma_j S_j).
\]
Since $\operatorname{Tr}(\rho S_j) = \operatorname{Tr}(\rho(I-E_j)) = 1 - \operatorname{Tr}(\rho E_j)$, we have
\[
1 - \operatorname{Tr}(\rho_j E_j) = \alpha_j (1 - \operatorname{Tr}(\sigma_j E_j)).
\]
Rearranging this equation gives
\begin{equation}
\langle A_j, E_j \rangle = \operatorname{Tr}((\rho_j - \alpha_j \sigma_j) E_j) = \operatorname{Tr}(\rho_j E_j) - \alpha_j \operatorname{Tr}(\sigma_j E_j) = 1 - \alpha_j.
\end{equation}

Next, we show  that $\langle A_j, E_j \rangle \neq 0$, i.e., $\alpha_j \neq 1$. Suppose $\alpha_j = 1$. Then $\Phi_{\mathcal{E}_{\setminus j}}(\rho_j) = \Phi_{\mathcal{E}_{\setminus j}}(\sigma_j)$, and the calculation above shows $\operatorname{Tr}(\rho_j E_j) = \operatorname{Tr}(\sigma_j E_j)$. Combining these, we obtain $\Phi_{\mathcal{E}}(\rho_j) = \Phi_{\mathcal{E}}(\sigma_j)$. Since $\mathcal{E}$ is PSI-Complete (as it is vital), this implies $\rho_j = \sigma_j$. This contradicts Case 2, i.e.,  that $\rho_j$ and $\sigma_j$ are distinct. Therefore, $\alpha_j \neq 1$, and $\langle A_j, E_j \rangle \neq 0$.

\medskip
In both cases, for each $j$, we have constructed an operator $A_j$ in the operator space such that $\langle A_j, E_i \rangle = 0$ for $i \neq j$, and $\langle A_j, E_j \rangle \neq 0$. This implies that   $\{A_1, \ldots, A_m\}$ is a linearly independent set in the operator space, and therefore its size $m$ cannot exceed the dimension of the space. Thus, in the real case, $m \leq \dim \S_n = \binom{n+1}{2}$, and in the complex case, $m \leq \dim \H_n = n^2$.
\subsection{Proof that the Bounds are Sharp}

We now construct explicit vital rank-one POVMs achieving the maximum sizes. We describe these using their corresponding frames, recalling that the PSI-Complete and vitality properties are preserved under renormalization.

\begin{proposition} \label{prop.real_sharp}
The real frame $\mathcal{F}_{\R} = \left\{e_1, \ldots , e_n\right\} \cup \left\{e_i + e_j\right\}_{1 \leq i < j \leq n}$ in $\R^n$ is vital.
\end{proposition}
\begin{proof}
The size of the frame is $n + \binom{n}{2} = \binom{n+1}{2}$. We examine the corresponding measurement operators: $\{E_k = e_k e_k^T\}_{1 \leq k \leq n}$ and $\{E_{kl} = (e_k + e_l)(e_k + e_l)^T\}_{1 \leq k < l \leq n}$. It is known that these operators form a basis for $\S_n$ (see, e.g.,~\cite{balan2009painless}). Therefore, the measurement map $\Phi_{\mathcal{F}_{\R}}$ is injective on $\S_n$ (i.e., it is Informationally Complete), which implies the frame is PSI-Complete.

Since the measurement is IC and the size equals the dimension of $\S_n$ (Minimally IC), removing any vector $w$ results in a measurement map $\Phi_{\mathcal{F}_{\R} \setminus \{w\}}$ whose kernel is $1$-dimensional. To show vitality, we must show that for every $w$, this kernel is spanned by a rank $2$ operator   that can be written as $vv^T -  uu^T$.

\medskip
\noindent\textbf{Case 1: Delete $e_i$.}
Set
\[
v:=\sum_{m\ne i} e_m,\qquad  u:=v-2e_i,
\]
and define $
A:=vv^T-uu^T.
$
Since \(uu^T=(v-2e_i)(v-2e_i)^T = vv^T-2ve_i^T-2e_iv^T+4e_ie_i^T\), we have
\[
A=vv^T-uu^T = vv^T-\bigl(vv^T-2ve_i^T-2e_iv^T+4e_ie_i^T\bigr)
=2ve_i^T+2e_iv^T-4e_ie_i^T.
\]
Thus the entries of \(A\) are
\[
A_{rc}=\begin{cases}
-4, & r=c=i,\\[3pt]
2, & r=i,\ c\ne i,\\[3pt]
2, & r\ne i,\ c=i,\\[3pt]
0, & \text{otherwise},
\end{cases}
\]
and clearly $A$ has rank $2$.

Let $j\ne i$ then, 
\[
\langle A, E_j\rangle
=e_j^T A e_j = A_{jj}=0.
\]
Next, for  \(k\ne \ell\),
\[
\langle A,E_{k\ell}\rangle
=(e_k+e_\ell)^T A (e_k+e_\ell)
=A_{kk}+A_{\ell\ell}+2A_{k\ell}.
\]
There are two cases:

\smallskip
\emph{(1) Neither index is \(i\).} Then \(A_{kk}=A_{\ell\ell}=A_{k\ell}=0\), so \(\langle A,E_{k\ell}\rangle=0\).

\smallskip
\emph{(2) Exactly one index is \(i\).} WLOG take \(k=i\), \(\ell\ne i\). Then
\[
A_{ii}=-4,\qquad A_{\ell\ell}=0,\qquad A_{i\ell}=2
\]
and hence
\[
\langle A,E_{i\ell}\rangle = A_{ii}+A_{\ell\ell}+2A_{i\ell} = (-4)+0+2\cdot 2 = 0.
\]
The case \(k\ne i,\ \ell=i\) is identical by symmetry.

\medskip
Therefore, for this explicit rank-two symmetric matrix \(A\),
\[
\langle A,E_j\rangle=0\quad\text{for all }j\ne i,\qquad
\langle A,E_{k\ell}\rangle=0\quad\text{for all }k\ne \ell.
\]

\medskip
\noindent\textbf{Case 2: Delete $e_i+e_j$.}
We seek $A_{ij} \in \S_n$ orthogonal to $\{E_k\}_k$ and $\{E_{kl}\}_{(k,l)\neq(i,j)}$.
$\langle A_{ij}, E_k \rangle = (A_{ij})_{kk} = 0$ for all $k$.
$\langle A_{ij}, E_{kl} \rangle = (A_{ij})_{kk} + (A_{ij})_{ll} + 2(A_{ij})_{kl} = 2(A_{ij})_{kl} = 0$ for $(k,l)\neq(i,j)$.
Thus, $A_{ij}$ is zero everywhere except the $(i,j)$ and $(j,i)$ entries. Let $A_{ij} = e_i e_j^T + e_j e_i^T$. This matrix is indefinite (eigenvalues $\pm 1$) and has rank 2.

In both cases, the kernel contains a rank-2 indefinite operator, confirming the frame is vital.
\end{proof}
\begin{proposition} \label{prop.complex_sharp}
The complex frame $\mathcal{F}_{\C} = \left\{e_1, \ldots , e_n\right\} \cup \left\{e_k + e_l\right\}_{1 \leq k < l \leq n} \cup \left\{e_k + i e_l\right\}_{1 \leq k < l \leq n}$ in $\C^n$ is vital.
\end{proposition}
\begin{proof}
The size of the frame is $n + 2\binom{n}{2} = n^2$. The $n^2$ corresponding Hermitian measurement operators $ vv^*$ are known to form a basis for $\H_n$~\cite{balan2009painless}. Therefore, the measurement map $\Phi_{\mathcal{F}_{\C}}$ is injective, and the frame is PSI-Complete.

For vitality we need to show that for each $v \in \mathcal{F}_{\C}$, we construct an operator $A_v \in \ker \Phi_{\mathcal{F}_{\C} \setminus \{v\}}$ of rank $2$ that  can be written as $x x^* - y y^*$. We omit the details, as the arguments are very similar to the real case. 
\end{proof}

\subsection{Connection to Informational Completeness (IC)}
The following corollary shows that the POVMs constructed above are IC.
\begin{corollary}
The vital rank-one POVMs constructed in Proposition~\ref{prop.real_sharp} and Proposition~\ref{prop.complex_sharp} are Informationally Complete (IC).
\end{corollary}
\begin{proof}
In the proofs of both propositions, it was established that the measurement operators $\{v v^T\}_{v\in\mathcal{F}_{\R}}$ (resp. $\{v v^*\}_{v\in\mathcal{F}_{\C}}$) form a basis for the respective operator spaces ($\S_n$ or $\H_n$). By definition, this means the measurement map $\Phi_{\mathcal{F}}$ is injective on the entire space, and thus the measurement is IC.
\end{proof}

\begin{remark}
A measurement is called \emph{Minimally IC} if it is IC and its size equals the dimension of the operator space. The constructions above are therefore also minimally IC. 

However, it is not generally true that a Minimally IC measurement is vital in the PSI-Complete sense. This is because minimally IC measurement requires minimality with respect to the property of being injective on the entire operator space, whereas vitality in the sense of PSI-Complete requires minimality with regard to being injective only on the set of rank-$1$ operators. 
\end{remark}

\section{Construction of vital rank-one POVMs from Block Designs} \label{sec.design}
In this section, we provide a method for constructing vital rank-one POVMs in real Hilbert spaces using combinatorial designs.
This adds to an existing, fruitful line of research using algebraic designs, such as unitary $t$-designs, for characterizing quantum systems \cite{10.1063/1.2716992}.

Our analysis will diverge from the algebraic characterization given in 
Corollary \ref{stam_cor} and use a more concrete geometric condition that is specific to the real case, known as the  \emph{Complement Property} \cite{balan2006signal}. We will use this condition to prove that the constructed rank-one POVMs are vital. 
\begin{definition}[Complement Property]
A frame $\mathcal{F} \subset \R^n$ satisfies the \emph{Complement Property} if for every partition of $\mathcal{F}$ into two subsets $\mathcal{F}_1 \cup \mathcal{F}_2$, at least one of the subsets spans $\R^n$.
\end{definition}
The following result connects the complement property and  PSI-Completeness.
\begin{theorem}
\label{thm-psi-complete-complememt-property}
    [see \cite{balan2006signal}]
    A frame $\mathcal{F}=\{v_i\}_{i=1}^m \subset \R^n$ permits phase retrieval (equivalently,  its associated  POVM is PSI-Complete) if and only if it satisfies the Complement Property.
\end{theorem}
By Theorem \ref{thm-psi-complete-complememt-property}, we have the following characterization of vital rank-one POVMs in the real case
\begin{corollary}
    A frame $\mathcal{F}=\{v_i\}_{i=1}^m \subset \R^n$ permits phase retrieval and is vital (equivalently,  its associated  POVM is a vital rank-one POVM) if and only if it satisfies the Complement Property, but any subcollection of $\mathcal{F}$ does not.
\end{corollary}

\subsection{Block Designs and Associated Frames}

\begin{definition}[Block Design]
An $(n,k,\lambda)$-design (or $2$-$(n,k,\lambda)$ design) is a pair $(X, \mathcal{B})$, where $X=[n]=\{1,\ldots,n\}$ is a set of points, and $\mathcal{B} = \{B_1, \ldots, B_b\}$ is a collection of subsets of $X$ (called blocks), such that every block has size $k$, and every pair of distinct points is contained in exactly $\lambda$ blocks.
\end{definition}

We are interested in a specific class of designs known as quasi-residual designs (see, e.g.,~\cite{ColbournDinitz2006}), which correspond to the parameters $(n, w, w-1), w\geq 2$. Standard counting arguments show that in such a design, every point is contained in $r = n-1$ blocks, and the total number of blocks is $b = \frac{n(n-1)}{w}$. This implies that such designs exist only if $w$ divides $n(n-1)$.

We now define the family of frames associated with a block design.

\begin{cnst}[Frames from Block Designs]
\label{main-construction}
Let $\mathcal{B}$ be an $(n,k,\lambda)$-design. For each block $B_i$, let $\mathbb{A}_i \subset \R^n$ be the subspace of vectors whose support lies in $B_i$ ($\dim \mathbb{A}_i = k$). We define a family of frames parameterized by vectors $v_i \in \mathbb{A}_i$:
\[
\mathcal{F} = \{e_1, \ldots, e_n\} \cup \{v_1, \ldots , v_b\},
\]
where $\{e_j\}$ are the standard basis vectors.
\end{cnst}

This construction defines a family of $ kb$-dimensional frames. Next, we state and prove the main theorem of this section, showing that a generic choice of the vectors $v_i$ yields a vital frame, provided the underlying design is an $(n,w,w-1)$ design. 
Equivalently, the frame gives rise to a vital rank-one POVM. 

\begin{theorem} \label{thm.blockdesign}
Let $\mathcal{B}$ be an $(n,w,w-1)$ block design with $w\geq 2$. If the vectors $v_i \in \mathbb{A}_i$ in Construction \ref{main-construction} are chosen generically, then the resulting frame $\mathcal{F}$ is vital. Equivalently, it gives rise to a vital rank-one POVM of size $n + n(n-1)/w$.
\end{theorem}
\begin{remark}[Implications for Structured QST and Sparsity]
\label{remark:sparsity}Theorem~\ref{thm.blockdesign} provides a systematic method for designing efficient tomography experiments for Rebits (real quantum states). The construction reveals a trade-off, parameterized by $w$, between the total number of measurement outcomes ($m=n+n(n-1)/w$) and the complexity (support size $w$) of the structured measurements $\{v_i\}$.

When $w=n$ (Minimal Size, $m=2n-1$), the number of measurements is minimized, but the vectors $v_i$ are dense (full support). Conversely, when $w=2$ (Maximal Size, $m=\binom{n+1}{2}$), the vectors $v_i$ are 2-sparse. Such sparse measurements might be easier to implement physically. This framework offers a spectrum of structured, vital rank-one POVMs between these extremes, guaranteed by combinatorial properties.
\end{remark}
\subsection{Proof of Theorem~\ref{thm.blockdesign}}
To prove vitality, we will prove that the  Complement Property holds due to the combinatorial structure of the design. We will need the following two simple lemmas regarding the Complement Property in $\R^n$.

\begin{lemma} \label{lem.complement}
A frame $\mathcal{F} \subset \R^n$ satisfies the Complement Property if and only if for every nonzero $u\in\R^n$, the set
\[
\mathcal{F}_u := \{ v \in \mathcal{F} : \langle u, v \rangle \neq 0 \}
\]
spans $\mathbb{R}^n$.
\end{lemma}
\begin{proof}
($\Rightarrow$) Given $u \neq 0$, we partition $\mathcal{F}$ into $\mathcal{F}_u$ and $\mathcal{F} \cap u^\perp$. The latter set is contained in $u^\perp$, so it does not span $\R^n$. By the Complement Property, $\mathcal{F}_u$ must span $\R^n$.

($\Leftarrow$) Suppose $\mathcal{F}$ is partitioned into $\mathcal{F}_1$ and $\mathcal{F}_2$, and $\mathcal{F}_1$ does not span. Let $u \in (\operatorname{Span}\mathcal{F}_1)^\perp$ be nonzero. Then $\mathcal{F}_1 \subset u^\perp$. This implies $\mathcal{F}_2 \supset \mathcal{F}_u$. By hypothesis, $\mathcal{F}_u$ spans $\R^n$, so $\mathcal{F}_2$ also spans $\R^n$.
\end{proof}

\begin{lemma}
\label{prop1}
A frame $\mathcal{F} \subseteq \mathbb{R}^n$ satisfying the Complement Property is vital if there exists a basis $\{u_1, \dots, u_n\}$ of $\mathbb{R}^n$ such that for each $i$, the set $\mathcal{F}_{u_i}$ is a basis of $\mathbb{R}^n$, i.e., $|\mathcal{F}_{u_i}|=n$.
\end{lemma}
\begin{proof}
Consider any vector $v \in \mathcal{F}$. Since $v \neq 0$, there exists some $u_i$ such that $\langle u_i, v \rangle \neq 0$, so $v \in \mathcal{F}_{u_i}$. By hypothesis, $\mathcal{F}_{u_i}$ is a basis. Let $\mathcal{F}' = \mathcal{F} \setminus \{v\}$. Then $\mathcal{F}'_{u_i} = \mathcal{F}_{u_i} \setminus \{v\}$. This set has size $n-1$ and does not span $\R^n$. By Lemma~\ref{lem.complement}, $\mathcal{F} '$ does not satisfy the Complement Property. Thus, $\mathcal{F}$ is vital.
\end{proof}

We now analyze the structure of the frame constructed from an $(n,w,w-1)$ design in Construction \ref{main-construction}. Let $\mathcal{F} =\{e_1, \ldots, e_n, v_1, \ldots , v_b\}$. Consider the basis vectors $\{e_i\}$. For a given $e_i$, the set $\mathcal{F}_{e_i}$ contains $e_i$. It also contains $v_k$ if and only if $i$ is in the support of $v_k$. Assuming generic nonzero coefficients, this implies that  $i \in B_k$. In an $(n,w,w-1)$ design, each point is contained in $r=n-1$ blocks. Thus, $|\mathcal{F}_{e_i}| = 1 + (n-1) = n$.

If we establish that $\mathcal{F}$ satisfies the Complement Property, then by Lemma~\ref{lem.complement}, $\mathcal{F}_{e_i}$ must span $\R^n$. Since $|\mathcal{F}_{e_i}|=n$, it forms a basis. Consequently, by Lemma~\ref{prop1}, the frame $\mathcal{F}$ must be vital.

Therefore, to prove Theorem~\ref{thm.blockdesign}, it suffices to show that a generic frame $\mathcal{F}$ associated with the design satisfies the Complement Property. We reduce this requirement to a combinatorial condition on the design's incidence graph $G$ (the bipartite graph between points $[n]$ and blocks $\mathcal{B}$).

\begin{lemma}
\label{lemma-main}
Let $\mathcal{B}$ be a block design. For generic vectors $v_i \in \mathbb{A}_i$, the frame $\mathcal{F}$ satisfies the Complement Property if the incidence graph $G$ satisfies the following condition: for any disjoint nonempty subsets $S, T \subseteq [n]$ of points,
\begin{equation}
\label{eq:1}
|N(S)\cap N(T)| \geq |S| + |T| - 1,
\end{equation}
where $N(S)$ denotes the set of neighbors (blocks) of the set $S$.
\end{lemma}

\begin{proof}
We use the criterion of Lemma~\ref{lem.complement}. Let $u \in \R^n$ be nonzero. We need to show that $\mathcal{F}_u$ spans $\R^n$. Let $t$ be the size of the support of $u$. If $t=n$, then $\{e_1, \ldots, e_n\} \subset \mathcal{F}_u$, so it spans. Assume $1 \leq t \leq n-1$. Without loss of generality, assume $\supp(u) = [t]$.

Claim 1: At most $t-1$ vectors in $\mathcal{F}$ whose support intersects $[t]$ are orthogonal to $u$.
Suppose otherwise. Assume there exist $t$ distinct vectors $W=\{w_1, \ldots, w_t\} \subset \mathcal{F}$ such that $\supp(w_i) \cap [t] \neq \emptyset$ and $\langle w_i, u \rangle = 0$.
First, we observe that $W$ must consist entirely of block vectors $\{v_k\}$. If $w_i = e_j$, the condition $\supp(w_i) \cap [t] \neq \emptyset$ implies $j \in [t]$. But since $[t]$ is the support of $u$, we have $\langle w_i, u \rangle = u_j \neq 0$. This contradicts $\langle w_i, u \rangle = 0$.
Thus, $W \subset \{v_k\}_{k=1}^b$. Let $w_i'$ be the projection of $w_i$ onto $\mathbb{R}^{[t]}$, and $u'$ be the projection of $u$. Then $w_i' \neq 0$, $\langle w_i', u' \rangle = 0$, and $u'$ has full support in $\mathbb{R}^{[t]}$.
Let $S_i = \mathrm{Supp}(w_i')$, and since $S_i \subseteq [t]$, it holds that
$$\Bigl|\bigcup_{i=1}^t S_i\Bigr|\le t.$$
Since the vectors $\{v_k\}$ are chosen generically within their respective block supports, and thus also the $w_i$'s in $W$, the projections $\{w_i' \}$ are generic within the supports $\{S_i\}$.
To arrive at a contradiction, we will need the following lemma.
\begin{lemma}
\label{good-lemma}
Let $S_1,\ldots,S_t\subseteq [t]:=\{1,\ldots,n\}$ be nonempty sets such that
\[
\Bigl|\bigcup_{i=1}^t S_i\Bigr|\le t.
\]
For each $i$, let
\[
V_i:=\{v\in\mathbb{R}^t:\operatorname{Supp}(v)\subseteq S_i\}\cong \mathbb{R}^{|S_i|},
\qquad
\mathbb{A}:=\prod_{i=1}^t V_i.
\]
Then there is a nonempty Zariski open set $U\subset\mathbb{A}$ such that for every $(v_1,\ldots,v_t)\in U$, the subspace $\langle v_1,\ldots,v_t\rangle\subset\mathbb{R}^t$ contains at least one standard basis vector.
\end{lemma}
By Lemma~\ref{good-lemma}, the span of the $t$ vectors $w_i'$ contains at least one standard basis vector $e_j$ with $j \in [t]$. Since $\langle w_i', u' \rangle = 0$ for all $i$, it follows that $\langle e_j, u' \rangle = 0$. This contradicts the fact that $u'$ has full support. Thus, Claim 1 holds.

\textbf{Claim 2:} $\mathcal{F}_u$ spans $\R^n$.

We aim to identify $n$ vectors $w_1, \ldots , w_n \in \mathcal{F}_u$ such that $i \in \supp(w_i)$ for $i \in [n]$. Such a set of generic vectors forms a basis.

For $i=1, \ldots, t$, we set $w_i = e_i$, as $e_i \in \mathcal{F}_u$.

We need to find $w_i$ for $i=t+1, \ldots, n$ from the block vectors $\{v_j\}$. Let $\mathcal{B}_{int} = N([t])$ be the set of blocks intersecting the support of $u$. Let $\mathcal{B}' \subset \mathcal{B}_{int}$ be the subset of blocks whose corresponding vectors $v_j$ are in $\mathcal{F}_u$. By Claim 1, $|\mathcal{B}_{int} \setminus \mathcal{B}'| \leq t-1$.

We construct a matching between the remaining coordinates $\{t+1, \ldots, n\}$ and the blocks in $\mathcal{B}' $ by verifying that Hall's condition holds. Let $S \subseteq \{t+1, \ldots, n\}$. We must show $|N(S) \cap \mathcal{B}'| \geq |S|$.

By  \eqref{eq:1} with $T=[t]$ we have
\[
|N(S) \cap N([t])| \geq |S| + t - 1.
\]
We relate this to $\mathcal{B}'$:
\begin{align*}
|N(S) \cap \mathcal{B}'| &= |N(S) \cap N([t])| - |N(S) \cap (N([t]) \setminus \mathcal{B}')| \\
&\geq |N(S) \cap N([t])| - |N([t]) \setminus \mathcal{B}'| \\
&\geq (|S| + t - 1) - (t-1) = |S|.
\end{align*}
Thus, there exists a matching that provides the required vectors $w_{t+1}, \ldots, w_n$. Therefore, $\mathcal{F}_u$ spans $\R^n$.
\end{proof}

To complete the proof of Theorem~\ref{thm.blockdesign}, we need to show that an $(n,w,w-1)$-designs satisfy the combinatorial condition \eqref{eq:1} and prove Lemma \ref{good-lemma}. We begin with the proof of the lemma.

\begin{proof}[Proof of Lemma \ref{good-lemma}]
Let $U(T):=\bigcup_{i\in T}S_i$ for $T\subseteq [t]$. Since $\bigl|\bigcup_{i=1}^t S_i\bigr|\le t$, the family
\[
\mathcal{F}:=\{\, T\subseteq [t] : |U(T)|\le |T|, T\neq  \emptyset\,\}
\]
is nonempty. Let $T\in\mathcal{F}$ that is \emph{minimal by inclusion}. By minimality,
\begin{equation}\label{eq:strict}
|U(R)|>|R|\quad\text{for every nonempty proper }R \subsetneq T,
\end{equation}
since otherwise $R$ would violate minimality. In particular, if $|R|=|T|-1$, then $|U(R)|\ge |T|$. As $U(R)\subseteq U(T)$, we obtain $|U(T)|\ge |T|$. Together with $|U(T)|\le |T|$ (because $T\in\mathcal{F}$), this forces
\begin{equation}\label{eq:equality}
|U(T)|=|T|.
\end{equation}

Consider the bipartite graph with left vertices $T$, right vertices $J:=U(T)$, and edges $(i,j)$ whenever $j\in S_i$. By \eqref{eq:strict} and \eqref{eq:equality}, we have
\[
|U(R)|\ge |R|\quad\text{for all }R\subseteq T,
\]
(with strict inequality for proper $R\subsetneq T$ and equality when $R=T$). Hence \emph{Hall's condition} holds for the family $(S_i)_{i\in T}$ and, because $|J|=|T|$, there exists a bijection (perfect matching)
\[
\sigma:T\to J\quad\text{with}\quad \sigma(i)\in S_i\ \text{for all }i\in T.
\]

Now restrict coordinates to $J$ and form the $|J|\times|T|$ matrix $M$ whose $i$-th column (for $i \in T$) is $v_i|_J$. The determinant $\det(M)$ is a polynomial in the coordinate entries of $(v_i)_{i \in T}$. To show this polynomial is not identically zero, we demonstrate a specific specialization where it is nonzero. Using the matching, choose the specialization
$$v_i:=e_{\sigma(i)}\in V_i\quad (i\in T).$$

Then $M$ is the permutation matrix $(\delta_{j,\sigma(i)})_{j\in J,\ i\in T}$, so $\det(M)=\pm1\neq 0$. Since the polynomial $\det(M)$ is not identically zero, its nonvanishing locus
$$ U_T:=\{(v_i)_{i\in T}\in \textstyle\prod_{i\in T} V_i:\ \det(M)\neq 0\}
$$
is a nonempty Zariski open set on which the vectors $(v_i)_{i\in T}$ are linearly independent and thus $\langle v_i:\ i\in T\rangle=\mathbb{R}^J$, in particular containing each $e_j$ for $j\in J$.
Finally, set
\[
U:=U_T\times \prod_{i\notin T} V_i\subset\mathbb{A}.
\]
For any $(v_1,\ldots,v_t)\in U$, the span $\langle v_1,\ldots,v_t\rangle$ contains $\langle v_i:\ i\in T\rangle=\mathbb{R}^J$, and therefore contains at least one standard basis vector $e_j$ (indeed, all with $j\in J\subseteq \bigcup_{i=1}^t S_i$).
\end{proof}

Lastly, we show that $(n,w,w-1)$-designs satisfy the combinatorial condition \eqref{eq:1}.
\begin{proposition} \label{prop.combinatorial_condition}
Let $G$ be the bipartite incidence graph of an $(n,w,w-1)$-design. Then, for any two non-empty, disjoint point sets $S, T\subseteq[n]$,
\[
|N(S)\cap N(T)|\;\ge\;|S|+|T|-1 .
\]
\end{proposition}

\begin{proof}
Let $s=|S|$ and $t=|T|$. Assume without loss of generality $1\le t\le s$. We use a double counting argument based on pairs $(p,q)\in S\times T$. Every such pair is contained in exactly $w-1$ blocks. Therefore,
\begin{equation}\label{eq:pair-count}
(w-1)\,s\,t \;=\; \sum_{B\in N(S)\cap N(T)} |B\cap S|\cdot|B\cap T|.
\end{equation}
Since $S$ and $T$ are disjoint, $|B\cap S|+|B\cap T| \le |B| = w$. We analyze by cases.

\smallskip
\noindent\textbf{Case 1: $s+t<w$.}
Choose any $p\in S$ and $q\in T$. The pair $\{p,q\}$ lies in $w-1$ blocks. Thus,
$|N(S)\cap N(T)|\;\ge\;|N(p)\cap N(q)|=w-1\ge s+t-1$, since $s+t < w$.

\smallskip
\noindent\textbf{Case 2: $s+t\ge w$ and $t\le w/2$.}
In this case, the product $|B\cap S|\cdot|B\cap T|$ is at most $t(w-t)$.
Substituting this bound into~\eqref{eq:pair-count} gives,
\[
|N(S)\cap N(T)| \;\ge\; \frac{(w-1)\,s\,t}{t\,(w-t)} \;=\; \frac{s\,(w-1)}{w-t}.
\]
Define $g(t)=\frac{s(w-1)}{w-t}-t$. We analyze the behavior of $g(t)$ for $1 \le t \le w/2$.
Since $s+t \ge w$, we have $s \ge w-t$.
The derivative is $g'(t)=\frac{s(w-1)}{(w-t)^2}-1$.
\[
g'(t) \;\ge\;\frac{(w-t)(w-1)}{(w-t)^2}-1 \;=\;\frac{w-1}{w-t}-1.
\]
Since $t \ge 1$, $w-t \le w-1$, so $g'(t) \ge 0$.
Thus, $g(t)$ is increasing in $t$. Its minimum occurs at $t=1$.
$g(1) = \frac{s(w-1)}{w-1} - 1 = s-1$.
Therefore, $|N(S)\cap N(T)| \ge s+t-1$.

\smallskip

\noindent\textbf{Case 3: $s+t\ge w$ and $t>w/2$.}
The product $|B\cap S|\cdot|B\cap T|$ is maximized when the sizes are as close as possible, i.e., $\lceil w/2\rceil$ and $\lfloor w/2\rfloor$. 
Write $C:=\lceil w/2\rceil\lfloor w/2\rfloor$. Then, 
\eqref{eq:pair-count} yields
\[
|N(S)\cap N(T)| \;\ge\; \frac{(w-1)st}{C}.
\]
Set $h(s,t):=\frac{(w-1)st}{C}-(s+t)$. On the domain $s\ge t\ge
\lfloor w/2\rfloor+1$ we have
\[
\frac{\partial h}{\partial s}=\frac{w-1}{C}\,t-1\ge 0,\qquad
\frac{\partial h}{\partial t}=\frac{w-1}{C}\,s-1\ge 0.
\]
 Thus
$h$ is increasing in each variable, and its minimum on the domain
occurs at $s=t=\lfloor w/2\rfloor+1$. A direct computation gives
\[
h\bigl(\lfloor w/2\rfloor+1,\ \lfloor w/2\rfloor+1\bigr)=
\begin{cases}
1-\dfrac{1}{k^2} & \text{if } w=2k,\\[6pt]
0 & \text{if } w=2k+1.
\end{cases}
\]
Hence $h(s,t)\ge 0$, i.e. $|N(S)\cap N(T)|\ge s+t\ge s+t-1$, as needed.
\qedhere

\end{proof}

\subsection{Structurally Explicit Constructions and Explicit Examples}
Theorem~\ref{thm.blockdesign} provides a general framework for constructing vital rank-one POVMs. Whenever the underlying $(n,w,w-1)$-design is explicitly defined (as in Example \ref{ex:finite-fields}, below), the resulting frame construction can be considered \emph{structurally explicit}—that is, the supports of the measurement vectors $\{v_i\}$ are explicitly known. Theorem~\ref{thm.blockdesign} guarantees that a generic choice of coefficients within these supports yields a vital frame (provided $w\ge 2$).We now highlight three canonical families of designs that exist for all $n$ and illustrate the spectrum of trade-offs between size and sparsity discussed in Remark \ref{remark:sparsity}. In some of these cases, we can go beyond structural explicitness and provide \emph{fully explicit} constructions, meaning concrete, closed-form choices of the coefficients, rather than relying only on the genericity argument.

\begin{enumerate}\item \textbf{Minimal Size ($w=n$):} The design consists of $b=n-1$ blocks, where each $B_i = [n]$. The associated frames have size $2n-1$. The frame is of the form $\mathcal{F} = \{e_1, \ldots , e_n, v_1, \ldots , v_{n-1}\}$, where the $v_i$ have full support. This recovers the known result that frames of size $2n-1$ in $\R^n$ that permit phase retrieval are vital. \emph{Fully Explicit Construction:} For a frame of size $2n-1$, the Complement Property is equivalent to the frame being \emph{full spark} (i.e., every subset of $n$ vectors is linearly independent). This occurs if and only if the $n \times (n-1)$ matrix $V$ formed by the columns $\{v_j\}$ is \emph{super-regular} (i.e., every square submatrix is invertible). We can explicitly construct such a matrix $V$ using a Cauchy matrix structure, to obtain an explicit  
 frame $\mathcal{F}$ that is  full spark and vital.\item \textbf{Size $2n$ ($w=n-1$):} The design consists of $b=n$ blocks (assuming $n\ge 3$), where the blocks are $B_i = [n] \setminus \{i\}$. The associated frames have size $2n$. The frame is of the form $\mathcal{F} = \{e_1, \ldots, e_n, v_1, \ldots , v_n\}$, where $v_i$ is supported on $[n] \setminus \{i\}$. By Theorem~\ref{thm.blockdesign}, a generic choice of coefficients within these supports yields a vital frame. Constructing a fully explicit example for this case appears non-trivial. For instance, the natural choice $v_i = \sum_{j\neq i} e_j$ does not generally satisfy the Complement Property for $n\ge 4$. Indeed, for $n=4$, consider $u=e_1-e_2$, and the set $\mathcal{F}_u$ does not span $\R^4$.\item \textbf{Maximal Size ($w=2$):} The design consists of $b=\binom{n}{2}$ blocks, where the blocks are all pairs $\{i,j\}$. The associated frames have size $\binom{n+1}{2}$. The frame has the form $\mathcal{F} = \{e_1, \ldots , e_n\} \cup \{v_{ij}\}_{1 \leq i < j \leq n}$, where $v_{ij}$ is supported on $\{i,j\}$. This construction is equivalent (in the sense of Section \ref{sec:equiv-vital-measurements}, below) to the explicit maximal vital frame constructed in Proposition~\ref{prop.real_sharp}  which has the simplest coefficients: $v_{ij} = e_i + e_j, \text{for } 1 \leq i < j \leq n.$
 \end{enumerate}
We conclude this section with an example of a vital frame that is not covered by the constructions above, highlighting the flexibility of our framework. Unlike the previous constructions, it is currently unknown how to make these examples explicit.
\begin{example}
\label{ex:finite-fields}
We construct a vital frame in $\R^7$ of size $7 + \frac{7\cdot 6}{3} = 21$ using a $(7,3,2)$ design derived from the affine group $\mathrm{AGL}(1,7)$.

The construction of the design and its proof that it is indeed a design are given next. 
\begin{proposition}\label{prop.732}
Let $V=\mathbb{F}_7$ and $G=\mathrm{AGL}(1,7)=\{x\mapsto ax+b:\ a\in\mathbb{F}_7^\times,\ b\in\mathbb{F}_7\}$. Let $B_0=\{0,1,3\}\subset V$. The orbit of $B_0$ under the action of $G$, $\mathcal{B}=\{g(B_0): g\in G\}$, forms a $(7,3,2)$ block design.
\end{proposition}
\begin{proof}
We determine the stabilizer of $B_0$, $G_{B_0}$. The affine map $f(x)=2x+1\pmod 7$ satisfies $f(0)=1, f(1)=3, f(3)=0$. Thus $f$ cyclically permutes $B_0$, so $|\langle f\rangle|=3$ and $\langle f\rangle \subseteq G_{B_0}$. It can be verified that $|G_{B_0}|=3$.

The size of the group is $|G|=42$. By the orbit-stabilizer theorem, the number of blocks is $|\mathcal{B}|=42/3=14$. Each block has a size of $3$.

Since the action of $\mathrm{AGL}(1,7)$ is 2-transitive on $\mathbb{F}_7$, the action on pairs of points is transitive. Therefore, every pair is contained in the same number of blocks, $\lambda$. We calculate $\lambda$ by the standard identity $b\binom{k}{2} = \lambda \binom{n}{2}$:
\[
14 \binom{3}{2} = \lambda \binom{7}{2} \implies 42 = 21\lambda \implies \lambda=2.
\]
\end{proof}
\end{example}

\section{Equivalence of vital rank-one POVMs}
\label{sec:equiv-vital-measurements}
We now consider the classification of vital rank-one POVMs. The property of being vital depends on the geometric configuration of the measurement vectors. We formalize the notion of equivalence between measurements.

If $\mathcal{F} = \{v_1, \ldots , v_m\}$ is a frame, applying an invertible linear transformation $A \in \GL_n(\K)$ to all vectors yields a new frame $A\mathcal{F} = \{Av_1, \ldots, Av_m\}$. This transformation preserves the property of phase retrieval (i.e., PSI-Completeness) and vitality. Furthermore, the ordering of the vectors does not affect these properties. 

\begin{definition}
Two $m$-element frames $\mathcal{F}_1, \mathcal{F}_2$ in $\K^n$ are \emph{equivalent} if there exists an element $A \in \GL_n(\K)$ and a permutation $\sigma \in S_m$ such that $A\mathcal{F}_1 = \sigma \mathcal{F}_2$ as sets of vectors.
\end{definition}

\begin{remark}[Physical Equivalence in QST]
In the context of Quantum State Tomography, equivalence is often defined in terms of unitary transformations $U(n)$ (or orthogonal transformations $O(n)$ in the real case), as these preserve the underlying physics and the normalization of the states. We adopt the broader $\GL_n$ equivalence here, which is standard in frame theory, to classify the fundamental geometric arrangement of the measurement subspaces, independent of normalization or the specific choice of basis defining the inner product.
\end{remark}

The following proposition lets us assume that any frame has a standard form.

\begin{proposition} \label{prop.equiv}
Every $m$-element frame $\mathcal{F}$ in $\K^n$ is equivalent to a frame of the form $$\{e_1, \ldots , e_n, v_1, \ldots, v_{m-n}\},$$ where $\{e_1, \ldots , e_n\}$ is the standard basis.
\end{proposition}
\begin{proof}
Let $\mathcal{F} = \{w_1, \ldots , w_m\}$. Since $\mathcal{F}$ is a frame, it spans $\K^n$ and must contain a basis. After reordering, we may assume $\{w_1, \ldots , w_n\}$ is a basis. Let $A \in \GL_n(\K)$ be the transformation defined by $A w_i = e_i$ for $i =1, \ldots, n$. Applying $A$ to $\mathcal{F}$ gives the desired form.
\end{proof}

\begin{remark}
The set of equivalence classes of $ m$-element frames in $\K^n$ is isomorphic to the quotient of the Grassmannian $Gr(n,m)$ by the natural action of the symmetric group on $\K^m$.
\end{remark}

We now consider the question of whether every vital frame is equivalent to one coming from an $(n,w,w-1)$
block design for some $w | n(n-1)$. 
Obviously,
a vital frame of size not equal to $n + \frac{n(n-1)}{w}$ for some $w | n(n-1)$ cannot come from a block design. However, it is currently unknown whether such frames exist.

When $w= n$ then every generic frame is vital and they are all equivalent to one coming from an $(n,n,n-1)$ block design.
 Next, we show that for every $n$, there are families of vital frames of size $2n$ that are not equivalent to those from $(n,n-1,n-2)$ block designs.

\subsection{Non-equivalence for size \texorpdfstring{$2n$}{2n}}
To distinguish different families of frames, we analyze the combinatorial structure of their linear dependencies. Since equivalence under the action of $\GL_n$ preserves linear dependence, the structure of maximal non-spanning subsets, including their sizes and intersection patterns, is invariant under $\GL_n$ equivalence.

We will need the following standard result, whose proof is omitted.

\begin{lemma} \label{lem.generic_transversal}
Let $\Sigma_1, \ldots, \Sigma_n \subseteq [n]$. A set of vectors $\{w_i\}_{i=1}^n\subset \R^n$, chosen generically such that $\supp(w_i) \subseteq \Sigma_i$, spans $\R^n$ if and only if the $\ Sigma_i$'s have a perfect matching, i.e., there exists a permutation $\sigma$ such that $\sigma(i) \in \Sigma_i$ for all $i$.
\end{lemma}

We now characterize the non-spanning subsets of size $n$ for the block design construction corresponding to $w=n-1$.

\begin{lemma} \label{lem.nonspanning}
Let $\mathcal{F}_{B}$ be a generic frame constructed from an $(n,n-1,n-2)$ block design, i.e., $\mathcal{F}_{B} = \{e_1, \ldots, e_n, v_1, \ldots , v_n\}$, where $v_i$ is generically supported on $B_i=[n] \setminus \{i\}$.
Then any non-spanning subset $S \subset \mathcal{F}_{B}$ of size $n$ must have the form
$S_i = \{e_1, \ldots , \widehat{e_i}, \ldots, e_n, v_i \}$ for some $i \in [n]$.
\end{lemma}
Note that $\widehat{e_i}$ means that $e_i$ is omitted from the set $S_i.$
\begin{proof}
Let $S \subset \mathcal{F}_B$ have size $n$. By Lemma~\ref{lem.generic_transversal}, $S$ spans $\R^n$ if and only if the family of supports $(\Sigma_w)_{w\in S}$ have a perfect matching.

Let $S_E = S \cap \{e_i\}$ and $S_V = S \cap \{v_i\}$. Let $k=|S_V|$, so $|S_E|=n-k$.
Let $I_E$ be the set of indices of $S_E$. A perfect matching exists if and only if we can simultaneously match $S_E$ to $I_E$ and $S_V$ to the remaining coordinates $C = [n] \setminus I_E$. Note $|C|=k$. The matching of $S_E$ to $I_E$ is trivial. Simply, match $e_j$ to $j$. Next, we need to verify if $S_V$ can be matched to $C$.

\textbf{Case 1: $k=0$.} $S_V=\emptyset, C=\emptyset$. A matching exists. $S$ spans.

\textbf{Case 2: $k=1$.} $S_V=\{v_i\}$. $C=\{j\}$. A matching exists if and only if $j \in \supp(v_i) = B_i$, which means $j \neq i$.
If $j=i$, then $I_E = [n]\setminus\{i\}$, so $S=\{v_i\} \cup \{e_l\}_{l\neq i} = S_i$. No matching exists, and $S_i$ does not span.
If $j\neq i$, a matching exists. $S$ spans.

\textbf{Case 3: $k\geq 2$.} We verify Hall's condition for matching $S_V$ to $C$. Let $J \subseteq S_V$. We must show that the set of available coordinates $N(J) = (\bigcup_{v_j \in J} \supp(v_j)) \cap C$ satisfies $|N(J)| \geq |J|$.

If $|J| \geq 2$. Let $v_a, v_b \in J$ be distinct vectors. Then, the union of their supports is $B_a \cup B_b = [n]$.
Thus, $N(J) = [n] \cap C = C$.
Since $|J| \leq |S_V| = |C|$, we have $|N(J)| = |C| \geq |J|$.

If $|J|=1$. $J=\{v_i\}$. We need $|N(J)| \geq 1$.
$N(J) = B_i \cap C$, which implies that $C\backslash \{i\}\subseteq N(J)$ and therefore, $1\leq |C|-1\leq |N(J)|$ since $|C|\geq 2.$
Therefore, if $k\geq 2$, a matching always exists, and thus $S$ spans.

We conclude that the only non-spanning subsets of size $n$ are the sets $S_i$.
\end{proof}

We now prove the main result of this section.

\begin{proposition} \label{prop.non_equivalence}
If $n > 3$, there exist families of vital frames of size $2n$ in $\R^n$ which are not equivalent to frames coming from an $(n,n-1,n-2)$ block design.
\end{proposition}
\begin{proof}
We show that that the block design frame $\mathcal{F}_B$ and the construction of vital frames of size $2n$ given in~\cite{gonzalez2024thesis}, which we denote by $\mathcal{F}$ are not equivalent.
First, we recall the construction of the frame.  

The frame $\mathcal{F}$ has the form $\{e_1, \ldots , e_n, w_1, \ldots , w_n\}$, where $w_\ell$  is a generic vector supported on $[n] \setminus \{n-\ell+1\}$ for $1 \leq \ell \leq n-1$, and $w_n = \sum_{i=1}^{n-1} \lambda_i w_i$ with the $\lambda_i$'s chosen generically.

We analyze the maximal non-spanning subsets of the two frames. Since both $\mathcal{F}_B$ and $\mathcal{F}$ are vital frames in $\R^n$, they satisfy the Complement Property. If a subset $S$ is non-spanning, its complement $S^c$ must span $\R^n$. Thus $|S^c| \geq n$. Since $|\mathcal{F}_B|=|\mathcal{F}|=2n$, this implies  that necessarily, $|S| \leq n$. Therefore, the maximal non-spanning subsets have size at most $n$, and any non-spanning subset of size $n$ is maximal.

\textbf{Analysis of $\mathcal{F}$:}
Consider the subset $S_1 = \{e_1, \ldots , e_{n-1}, w_1\}$. The support of $w_1$ (i.e., when $\ell=1$) is $[n] \setminus \{n\}$. Thus, $S_1$ is a maximal non-spanning subset.
Consider the subset $S_2 = \{w_1, \ldots , w_n\}$. By construction, $w_n$ is a linear combination of $w_1, \ldots, w_{n-1}$, so $S_2$ is also a non-spanning maximal subset, and  $|S_1 \cap S_2| = |\{w_1\}|=1$.

\textbf{Analysis of $\mathcal{F}_{B}$:}
By Lemma~\ref{lem.nonspanning}, the only maximal non-spanning subsets are the sets $S_i$.
For any two distinct such subsets, $S_i$ and $S_j$, it holds that 
$|S_i \cap S_j|=|\{e_k\}_{k\neq i,j}|=n-2$.

We conclude that  $\mathcal{F}$ has two maximal non-spanning subsets with intersection size $1$, while any two distinct maximal non-spanning subsets in $\mathcal{F}_{B}$ have intersection size $n-2$. Since the structure of maximal non-spanning subsets is invariant under $\GL_n$ equivalence and  $n>3$ (equivalently $n-2 > 1$),  the frames are not equivalent.
\end{proof}

\begin{remark}
One can verify that, for $n=3$, the frames constructed in~\cite{gonzalez2024thesis} are indeed equivalent to those arising from a $(3,2,1)$ design.
\end{remark}
\section{Conclusions}
\label{sec:conclusions}
This paper introduced the concept of vital rank-one POVMs for Quantum State Tomography: rank-one POVMs that are minimally sufficient for reconstructing pure states (PSI-Complete). We established sharp upper bounds on the size of vital rank-one POVMs: $\binom{n+1}{2}$ for real Hilbert spaces and $n^2$ for complex Hilbert spaces, and provided explicit constructions attaining these bounds. Notably, these maximal vital constructions are also Minimally Informationally Complete, and are therefore capable of reconstructing any quantum state.

Our main result is the introduction of a framework that connects combinatorial block designs to the construction of vital rank-one POVMs in the real case. In particular, we proved that $(n,w,w-1)$-designs generate large families of vital rank-one POVMs of various sizes, by leveraging the Complement Property. Our framework unifies existing examples and provides new constructions, such as in Example \ref{ex:finite-fields}. Furthermore, we showed that these design-based constructions do not cover all possible constructions of vital rank-one POVMs, thereby highlighting the richness of the space of vital POVMs.

Several directions for future research emerge from this work.

\begin{itemize}
    \item \textbf{Stability and Noise Robustness:} vital rank-one POVMs, by definition, have no redundancy. In practical QST, redundancy often enhances stability against experimental noise. It would be valuable to analyze the noise robustness of the constructed vital rank-one POVMs, for instance, by studying their condition numbers or frame potentials, and comparing them to highly stable but redundant measurements like Mutually Unbiased Bases (MUBs).
    \item \textbf{Complex Combinatorial Constructions:} The methods in Section~\ref{sec.design} rely heavily on the Complement Property, which is specific to the real case. Developing analogous methods for constructing structured vital rank-one POVMs in complex Hilbert spaces remains an interesting open question.
    \item \textbf{Classification of vital rank-one POVMs:} As shown in Section \ref{sec:equiv-vital-measurements}, the classification of vital rank-one POVMs up to $\GL_n$ equivalence is non-trivial. Further investigation into the geometric and algebraic properties characterizing these equivalence classes is needed.
\end{itemize}
\section*{Acknowledgments.} D. Edidin and I. Gonzalez were partially supported by BSF grant 2020159 and NSF grant DMS2205626 while preparing this work.  I. Tamo was supported by the European Research Council (ERC) under Grant 852953.

\bibliographystyle{alpha}
\bibliography{ref}

\end{document}